\documentclass[a4paper]{scrartcl}

\usepackage{authblk}
\usepackage[T1]{fontenc}
\usepackage[utf8]{inputenc}
\usepackage{amsmath}
\usepackage{amssymb}
\usepackage{amsthm}
\usepackage{algpseudocode}
\usepackage{algorithm}
\usepackage{mathtools}
\usepackage{url}
\usepackage{tikz}
\usepackage{pgfplots}
\usetikzlibrary{positioning, pgfplots.statistics}

\newcommand{\R}{\mathbb{R}}
\newcommand{\IRMSD}{\mathit{IRMSD}}
\newtheorem{theorem}{Theorem}
\newtheorem{lemma}{Lemma}
\DeclareMathOperator*{\argmin}{arg\,min}

\newcommand{\norm}[1]{\left\Vert#1\right\Vert}

\date{}
\author[1,*]{Johannes Bulin}
\author[1,2]{Jan Hamaekers}
\affil[1]{Department of Virtual Material Design, Fraunhofer Institute
    for Algorithms and Scientific Computing, Schloss Birlinghoven, 53757
    Sankt Augustin, Germany
}

\affil[2]{Fraunhofer Center for Machine Learning, Schloss
   Birlinghoven, 53757 Sankt Augustin, Germany}
\affil[*]{Corresponding author: \texttt{johannes.bulin@scai.fraunhofer.de}}

\title{Similarity of particle systems using an invariant root mean square deviation measure}

\begin{document}
\maketitle

\begin{abstract}
Determining whether two particle systems are similar is a common problem in particle simulations.
    When the comparison should be invariant under permutations, orthogonal transformations, and translations of
    the systems, special techniques are needed.
    We present an algorithm that can test particle systems of finite size for similarity and, if they
    are similar, can find the optimal alignment between them.
    Our approach is based on an invariant version of the root mean square deviation (RMSD) measure and 
    is capable of finding the globally optimal solution in $O(n^3)$ operations where $n$ is the number 
    of three-dimensional particles.
\end{abstract}

\section{Introduction}

The comparison of particle systems is a subproblem that occurs frequently in many interesting problems ranging
from the analysis of certain substructures to different machine learning approaches.
One particular use case is databases that store particle systems and corresponding
properties.
In many cases, some of these properties such as the potential energy are invariant 
under particle index permutations, translations, and orthogonal transformations of the system.
Other properties, such as forces, can be readily transferred from one system to a transformed instance 
of itself if the corresponding transformations are known.
To exploit these invariances, it is necessary to check two particle systems for similarity up to
these invariances and -- if the two systems are similar -- to obtain the corresponding transformations.

Some attempts have been made that use this idea.
In \cite{jones2016}, a database was presented that stores atomic neighbourhoods and corresponding forces.
Instead of forces, the databases in \cite{mellouhi2008} and \cite{konwar2011} hold local transition information that is
then used in a localised kinetic Monte Carlo algorithm.
Depending on the choice of the comparison procedure, the problem of false positives and negatives may occur:
false positives (two systems are designated as similar, despite not being so) may have catastrophic 
consequences, as incorrect properties can be attributed to a system.
False negatives (two similar systems are not recognised as such) may cause a waste
of computational time, as properties that already exist in the database must be recomputed.
The comparison techniques in \cite{mellouhi2008} and \cite{konwar2011} may cause false positives
as the comparison is done on a reduced representation of a system.
In \cite{jones2016}, an approximate algorithm was used to check for similarity that can result in
false negatives.

To avoid false positives/negatives we introduce a comparison procedure that
is based on an invariant version of the root mean square deviation measure (RMSD).
It is able to check particle systems for similarity and -- if two systems are similar -- can
also calculate the transformations that align one system to the other.
We also prove that our algorithm finds the globally optimal solutions in no more than $O(n^3)$ operations
where $n$ is the number of three-dimensional particles.

\section{The invariant RMSD}

Before introducing the invariant RMSD, we need a mathematical definition of a \emph{particle system}.
In this work, we define a particle system $S = (\mathbf x, \mathbf e)$ by its $n$ $d$-dimensional coordinates 
$\mathbf x = (x_1,\dots, x_n)$ where $x_i\in\R^d$ and corresponding particle elements $\mathbf e = (e_1,\dots, e_n)$
with $e_i\in\mathbb{Z}$.
This allows us to represent molecules as well as atomic neighbourhoods.

For two particle systems $S = (\mathbf x, \mathbf e)$ and $\tilde S = (\mathbf{\tilde x}, \mathbf{\tilde e})$ the classical RMSD is 
defined as
\begin{equation}
    \mathit{RMSD}(S, \tilde S) = \begin{cases}
        \infty, & \exists i:\, e_i\neq \tilde e_i \\
        \sqrt{\sum_{i=1}^n \norm{x_i - \tilde x_i}_2^2}, & \textrm{otherwise}.
    \end{cases}
\end{equation}
Similar to the work in \cite{sadeghi2013, jones2016, li2007} we obtain invariance under orthogonal 
transformations, translations, and permutations by minimising over all 
such transformations.
We call the resulting distance measure the \emph{invariant RMSD (IRMSD)}:
\begin{align}
    \IRMSD(S, \tilde S) =     
        \min_{\pi\in \mathbb{S}_n, \mathbf{R}\in\mathbb{O}_d, \mathbf{T}\in\R^d} \begin{cases}
        \infty, & \exists i: \tilde e_{\pi_i} \neq e_i \\
        \sqrt{\norm{x_i - \mathbf{R}\tilde x_{\pi_i} - \mathbf{T}}_2^2}, & \textrm{otherwise}.
    \end{cases} \label{eq:IRMSD} 
\end{align}
Here, $\mathbb{O}_d$ is the set of all $d\times d$ orthogonal matrices while $\mathbb{S}_n$ is the
set of all permutations of length $n$.

The $\IRMSD$ has a number of properties that makes it suitable for comparing particle systems.
$\IRMSD(S,\tilde S) = 0$ if and only if $S$ and $\tilde S$ are equal up to the 
mentioned transformations.
Thus, if we define two systems as similar if $\IRMSD(S,\tilde S) = 0$ we will not have any 
problems with false positives or negatives.
$\IRMSD$ is also continuous with respect to the particle coordinates and -- if the elements
are ignored -- fulfils all properties of a pseudometric on $\R^{n\times d}$.

Due to the presence of numerical inaccuracies and different types of noise, checking for $\IRMSD(S,\tilde S) = 0$ will
in most cases not be useful. 
For practical purposes, we consider two particle systems $S$ and $\tilde S$ as similar if
\begin{equation} \label{eq:tolerance}
    \IRMSD(S,\tilde S)\leq\epsilon
\end{equation} 
for some tolerance $\epsilon$.

\section{Calculating the invariant RMSD}

For arbitrary particle systems $S$ and $\tilde S$, calculating $\IRMSD(S, \tilde S)$ 
is an NP-hard problem \cite{sadeghi2013, li2007}.
Nevertheless, several ways to calculate or approximate it have been proposed.
One class of algorithms calculates only approximations to $\IRMSD(S,\tilde S)$.
This includes the algorithms in \cite {sadeghi2013, jones2016} but also
various approaches in computer vision \cite{besl1992, gold1998} where point sets 
are compared, instead of particle systems.
As only approximations to the $\IRMSD$ are calculated, this class of algorithms may cause false positives/negatives 
when comparing particle systems.

A second class of algorithms, including the methods given in \cite{li2007, yang2013, griffiths2017}, is able
to calculate the exact value $\IRMSD(S,\tilde S)$ if run long enough (which may scale exponentially with
the number of particles). 
For many particle systems, some of these algorithms can nevertheless calculate $\IRMSD(S,\tilde S)$ rapidly.
For other systems, especially highly symmetric ones, they can be prohibitively
slow (see evaluation section).

In our setting it is sufficient to check whether $\IRMSD(S, \tilde S)\leq\epsilon$; only in this
case must we calculate the exact value of $\IRMSD$ and the corresponding transformations.
In the context of particle simulations, we can also assume that the distance between
two particles is bounded from below by a minimum particle distance $\mu$.
This can be motivated by the strong repulsive forces that act between atoms
that are close to each other.
Using this property, we introduce a new algorithm for this special kind of problem.
It can determine whether the inequality in equation \ref{eq:tolerance} holds true or not and, if so,
is able to calculate the exact value $\IRMSD(S,\tilde S)$ and the corresponding optimal
transformations.

In the following, 
if we compare two particle systems $S = (\mathbf x, \mathbf e)$ and $\tilde S = (\mathbf{\tilde x}, \mathbf{\tilde e})$ we assume that at least
one of them has full rank.
In this context that means that either $\mathbf x$ or $\mathbf{\tilde x}$ contain $d$ linearly independent points.
If this is not the case, one can reduce the dimensionality of the points (for example by using PCA) and calculate
the $\IRMSD$ for the reduced coordinates.

For notational simplicity, we skip the calculation of the optimal translations in equation \ref{eq:IRMSD}.
This can be done by pre-shifting both particle systems such that the centroids of the coordinates are both at the origin. 
In this case, the optimal translation is always the 0-vector.

\begin{algorithm}[ht!]
    \caption{Check whether $\IRMSD(S,\tilde S)\leq\epsilon$.}
    \label{alg:IRMSD}
    \begin{algorithmic}[1]
        \Require Two particle systems $S = (\mathbf x, \mathbf e)$, $\tilde S= (\mathbf{\tilde x},\mathbf{\tilde e})$.
        \Require $\mathbf{\tilde x}$ must have full rank.
        \Require $\epsilon < \frac{\mu}{2\sqrt{1 + 4d}}$.
        \Require $\norm{x_i - x_j}_2\geq \mu\; \forall i\neq j$. 
        \Require $\norm{\tilde x_i - \tilde x_j}_2\geq \mu\; \forall i\neq j$. 

        \State $g^\ast = \infty$
        \State $\mathbf R^\ast = I$
        \State $\pi^\ast = \mathit{id}$

        \If{$\mathit{sort}(e)\neq \mathit{sort}(\tilde e)$}
            \State \Return $\IRMSD(S, \tilde S)\not\leq\epsilon$
        \EndIf

        \State Choose the $d$ indices $j_1,\dots, j_d$ that maximise 
            the absolute value of the determinant of $[\tilde x_{j_1}, \dots, \tilde x_{j_d}]$. \label{line:indices}

        \ForAll{$i_1,\dots, i_d:\, e_{i_k} = \tilde e_{j_k}$ ($k=1,\dots, d$)}
            \State $\mathbf R = \argmin_{\mathbf R\in\mathbb{O}_d} \sum_{k=1}^d \norm{x_{i_k} - \mathbf R\tilde x_{j_k}}_2^2$ \label{line:R1}
            \State $g = \sum_{k=1}^d \norm{x_{i_k} - \mathbf R\tilde x_{j_k}}_2^2$ \label{line:dist1}
            \If{$g > \epsilon^2$}
                \State continue
            \EndIf
            \State Set $\pi\in \mathbb S_n$ s.t. $\pi_{i_k} = j_k$ ($\forall k\in\{1,\dots, d\}$), $\pi_k= \argmin_{l:\, e_k=\tilde e_l} \norm{x_k - \mathbf R\tilde x_l}_2$ ($\forall k\not\in\{i_1,\dots, i_d\})$ \label{line:pi}
            \State $\mathbf{\tilde R} = \argmin_{\mathbf R\in\mathbb{O}_d} \sum_{k=1}^n \norm{x_k - \mathbf R\tilde x_{\pi_k}}_2^2$ \label{line:r}
            \State $\tilde g = \sqrt{\sum_{k=1}^n \norm{x_k - \mathbf{\tilde R}\tilde x_{\pi_k}}_2^2}$ \label{line:dist2}
            \If{$\tilde g\leq g^\ast$}
                \State $g^\ast = \tilde g$
                \State $\mathbf R^\ast = \mathbf{\tilde R}$
                \State $\pi^\ast = \pi$
            \EndIf
        \EndFor

        \If{$g^\ast\leq \epsilon$}
            \Return $g^\ast, \mathbf R^\ast, \pi^\ast$
        \Else
            \State \Return $\IRMSD(S, \tilde S)\not\leq\epsilon$
        \EndIf
    \end{algorithmic}
\end{algorithm}

Our new algorithm (see algorithm \ref{alg:IRMSD}) is essentially a three-step procedure.
In the first step (line \ref{line:indices}), $d$ maximally linearly independent particles are chosen from the first 
system $S$, meaning that the corresponding particle coordinates maximise the determinant of the matrix
that they form.
Afterwards, all length-$d$ combinations of particles in the other system $\tilde S$ are determined such that these 
particles can be mapped onto the previously chosen $d$ particles in the other system, using only an orthogonal transformation
(up to the tolerance $\epsilon$).
If $\IRMSD(S,\tilde S)\leq \epsilon$ and $\epsilon < \frac{\mu}{2\sqrt{1 + 4d}}$,
 one of these orthogonal transformations (line \ref{line:R1}) is close enough to the globally 
optimal orthogonal transformation such that it can be used to determine
the globally optimal permutation (see theorem \ref{thm:algproof} in the appendix).
Thus, in one of the for-loop cycles algorithm \ref{alg:IRMSD} will calculate the globally optimal
permutation $\pi^\ast$ and thus the corresponding globally optimal orthogonal transformation in line \ref{line:r}.
All necessary proofs can be found in the appendix.

\section{Time complexity}

To determine the time complexity of algorithm \ref{alg:IRMSD}, its three key components must 
be analysed:
the cost of calculating the indices in line \ref{line:indices}, the cost of the lines \ref{line:R1}-\ref{line:dist1},
and the cost of the lines \ref{line:pi}-\ref{line:dist2}.

We begin with line \ref{line:indices}.
The $d$ indices that maximise the determinant can be calculated by a brute-force search
over all $\binom{n}{d}$ combinations.
For each combination the determinant of the corresponding matrix must be calculated, which requires $O(d^3)$ operations.
Combined, the time complexity of line \ref{line:indices} is $O(n^d d^3)$.

Calculating $\mathbf R$ in line \ref{line:R1} can be done by using the Kabsch algorithm \cite{kabsch1976}, which needs
$O(d^3)$ operations.
The subsequent calculation of $g$ in line \ref{line:dist1} also has time complexity $O(d^3)$.
As these two lines may be called up to $\binom{n}{d}$ many times, they contribute $O(n^d d^3)$ operations
to the total complexity.

It can be shown that there are no more than $2.415^d\binom{n}{d-1}$ index combinations $i_1,\dots, i_d$ such that
$g$ in line \ref{line:dist1} fulfils $g\leq\epsilon^2$ (see theorem \ref{thm:compproof} in the appendix).
Thus line \ref{line:pi} and subsequent parts of the code are called no more
than $O(2.415^d n^{d-1})$ times.
The permutation $\pi$ in line \ref{line:pi} can be calculated by a cell based approach
in $O(3^d d n)$ operations \cite{bentley1977}.
Both $\mathbf{\tilde R}$ and $\tilde g$ in lines \ref{line:R1} and \ref{line:dist2} can be calculated 
in $O(nd^2 + d^3)$ operations, once again using the Kabsch algorithm.
Therefore, the total 
contributions of lines \ref{line:pi} - \ref{line:dist2} to the time complexity is of order 
$O(2.415^dn^{d-1} (3^d dn + nd^2 + d^3)) = O(7.245^d n^d d)$.
Thus for the fixed ``standard'' dimensions $d = 2$ and $d = 3$, the algorithm has complexity
$O(n^2)$ and $O(n^3)$, respectively.

\section{Evaluation}

To verify the theoretical results we have implemented our algorithm for three-dimensional particles
as an extension module for the QuantumATK software package \cite{atk, smidstrup2019}.

Four different datasets were used for testing.
The first dataset (\emph{silicon dataset}) was derived from the database of silicon configurations in \cite{libatom, bartok2018}.
For each configuration in the database we calculated the atomic neighbourhoods of radius $6$Å of
all particles in this configuration.
Using these neighbourhoods, one million neighbourhood pairs $(S_i, \tilde S_i)$ were randomly chosen
such that $S_i$ and $\tilde S_i$ have the same number of particles. 

The next dataset (\emph{C720 dataset}) was generated using the C\textsubscript{720} fullerene molecule.
The coordinates $\mathbf c = (c_1,\dots, c_n)$ and elements $\mathbf e = (e_1,\dots, e_n)$ ($n=720$) of this molecule were taken and 1000 pairs
$(S_i,\tilde S_i)$ defined by 
\begin{equation}
\begin{aligned}
    S_i &= (\mathbf R^{(i)} \mathbf c_{\pi^{(i)}}, \mathbf e_{\pi^{(i)}}) \\
    \tilde S_i &= \left(\mathbf{\tilde R}^{(i)} \mathbf c_{\tilde \pi^{(i)}} + \frac{\tau^{(i)}}{\norm{\tau^{(i)}}_F} \rho^{(i)}, \mathbf  e_{\tilde \pi^{(i)}}\right)
\end{aligned} \label{eq:generate}
\end{equation}
were generated. $\mathbf R^{(i)}$ and $\mathbf{\tilde R}^{(i)}$ are random rotation matrices whereas $\pi^{(i)}$ and $\tilde \pi^{(i)}$ are
random permutations.
Additional noise was added in the form of $\tau^{(i)}$ which was drawn from the $3n$-dimensional multivariate standard normal distribution $\mathcal N(0, I)$.
This noise was scaled by $\rho^{(i)} / \norm{\tau^{(i)}}_2$ where $\rho^{(i)}$ is exponentially distributed with parameter $\lambda = 3$. 
This ensures that the magnitude of the noise
\begin{equation}
    \norm{\frac{\tau^{(i)}}{\norm{\tau^{(i)}}_F} \rho^{(i)}}_F = \rho^{(i)}
\end{equation}
varies from zero noise ($\rho^{(i)} = 0$) to very noisy ($\rho^{(i)}$ large).

The \emph{diamond dataset} was created in a similar way.
For different cutoff radii $r$ we calculated the atomic neighbourhood of some particle in a perfect
diamond crystal (the choice of the particle is unimportant as all neighbourhoods are equal).
Using the coordinates $\mathbf{c}^{(r)}$ and elements $\mathbf{e}^{(r)}$ of this neighbourhood we 
generated 1000 pairs $(S^{(r)}_i, \tilde S^{(r)}_i)$ in the same way as in 
equation \ref{eq:generate}.

Finally the \emph{spherical dataset} was created to investigate the behaviour of the algorithm 
over systems of different sizes.
For different numbers of particles $n$ we distributed $n$ points almost uniformly on the unit sphere 
using generalised spiral points from \cite{saff1994}.
The sphere was then scaled such that the minimal distance between two points became $2\textrm{Å}$.
These coordinates $\mathbf{c}^{(n)}$ were then used to build 100 particle system pairs $(S_i^{(n)}, \tilde S_i^{(n)})$
by the application of random orthogonal transformations and permutations (no noise was used).
For all particles, the same elements were used.

As reference algorithms we chose Go-ICP \cite{yang2013, goicp} and Go-Permdist \cite{griffiths2017, gopermdist}. 
These two algorithms are branch-and-bound based and calculate the optimal solution.
In both algorithm we have disabled the calculation of the optimal translation
as it caused Go-ICP to become much slower (the translation search in Go-ICP is expensive as it
also supports registration problems without one-to-one correspondences which is harder).
As mentioned earlier, it is also easy to decouple the calculation of the optimal translation
from finding the optimal permutation and orthogonal transformation.

We also modified the code such that both algorithm will stop their calculations if the lower
bounds in the branch-and-bound section become larger than the tolerance $\epsilon$. 
As we only need to check for similarity, it is not necessary to calculate the exact value $\IRMSD(S,\tilde S)$ 
once it has been established that $\IRMSD(S, \tilde S) > \epsilon$.
For both algorithm we stopped the calculation as soon as the difference between the 
lower and upper bounds became smaller than $10^{-2}$.

During our calculations we found a small mistake in the calculation of the lower bounds 
in the Go-Permdist algorithm that caused false negatives in a few cases.
We thus modified the Go-Permdist algorithm such that it used our corrected lower bounds.
The corrected terms can be found in the appendix.

\subsection{Accuracy}

We began by testing the accuracy of the algorithms.
To do this, we applied the three algorithms to all particle system pairs in the three datasets
and compared the results for different tolerances $\epsilon < \frac{\mu}{2\sqrt{1 + 4d}}$.
As Go-ICP and Go-Permdist optimise only over all rotations and not over all orthogonal transformations
(which also includes reflections), we encountered a few cases where our algorithm calculated
smaller $\IRMSD$ values.
This was expected and occurred only when the optimal orthogonal transformation was not a rotation
matrix. 
In the subsequent performance analysis we simply excluded these cases as not relevant.

Otherwise, we could not detect any differences between the three algorithms in terms of
the results that they calculated. 
The $\IRMSD$ values that the three algorithms computed never differed by more than $10^{-2}$, which was the
chosen stopping criterion for Go-ICP and Go-Permdist.

When testing the unfixed Go-Permdist algorithm we encountered some
false negatives where the algorithm wrongly determined that $\IRMSD(S,\tilde S) > \epsilon$.
The corrected Go-Permdist algorithm did not suffer from such problems.

\subsection{General Performance}

Next, we investigate the performance of our algorithm in terms of 
computational efficiency.
To do this we measured the wall time that was necessary to compare each pair in our three
datasets. 
All the tests were run on a single core of an Intel Xeon Gold 5118 processor using the tolerance
$\epsilon=0.2\textrm{Å}$.

\begin{figure}[ht!]
\begin{center}
\begin{tikzpicture}
\begin{semilogyaxis}[
    boxplot/draw direction = y,
    ylabel={Wall time in s},
    xtick={1,2,3},
    xticklabels = {Algorithm \ref{alg:IRMSD}, Go-ICP, Go-Permdist},
    xticklabel style = {align=center},
    ]

    \addplot+ [
        mark=*,
        black,
        mark options = {solid, fill=black},
        boxplot prepared={
            lower whisker = 0.00055,
            lower quartile = 0.0008,
            median = 0.00085,
            upper quartile = 0.00094,
            upper whisker = 0.075,
        },
    ] coordinates {(1, 0.002)};
    \addplot+ [
        mark=diamond*,
        blue,
        mark options = {solid, fill=blue},
        boxplot prepared={
            lower whisker = 0.000095,
            lower quartile = 0.0236,
            median = 0.0263,
            upper quartile = 0.0311,
            upper whisker = 2.301,
        },
    ] coordinates {(2, 0.034)};
    \addplot+ [
        mark=square*,
        red,
        mark options = {solid, fill=red},
        boxplot prepared={
            lower whisker = 0.000143,
            lower quartile = 0.00085,
            median = 0.00106,
            upper quartile = 0.0015,
            upper whisker = 371.8,
        },
    ] coordinates {(3, 0.00835)};
\end{semilogyaxis}
\end{tikzpicture}
\end{center}
\caption{Distribution of the wall times in the silicon dataset. The lines show (starting from below)
the 0\% (minimum), 25\%, 50\% (median), 75\% and 100\% (maximum) percentile of the wall times. The mean
runtime is represented by the single point.}
\label{fig:walltime_silicon}
\end{figure}
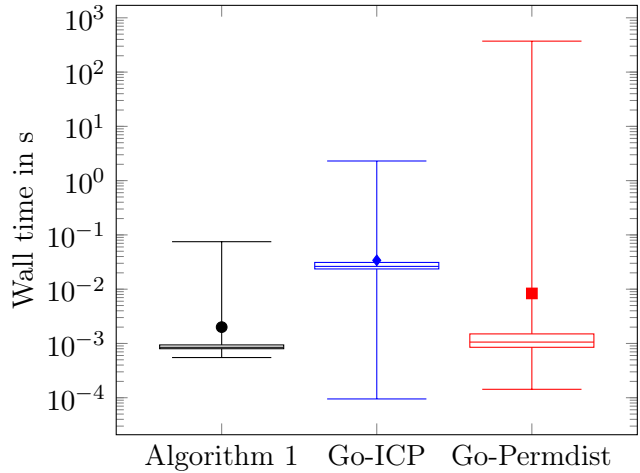

Figure \ref{fig:walltime_silicon} shows the distribution of the runtimes on the silicon dataset.
We observe that our algorithm is on average about four times faster than Go-Permdist and ten times faster than Go-ICP.
Interestingly, Go-ICP and Go-Permdist show a much higher variation in runtime.
Whereas our algorithm never needs more than 0.07s for a comparison, this value is
2.3s for Go-ICP and 371.8s for the (fixed) Go-Permdist algorithm.
We investigated this behaviour and observed that all three algorithms are slowest when the particle 
systems under investigation are highly symmetric.
We note that our algorithm is much faster than Go-ICP
and Go-Permdist in these cases.
In the silicon dataset this occurred when atomic neighbourhoods with (almost) perfect crystal 
structure were encountered.
By contrast, all three algorithms were very fast over particle systems with low symmetry
such as linear chains of atoms.

We turn now to the C720 dataset.
The C\textsubscript{720} fullerene is a highly symmetric molecule, so if our initial observations hold true, 
we might expect relatively slow calculations.
As shown in figure \ref{fig:walltime_c720} this seems to be the case.
Our algorithm is much faster here than the two other algorithms.
Go-ICP takes on average 
about 50 times more computational time, for Go-Permdist this figure increases to 500.
Once again, the runtimes of Go-ICP and Go-Permdist span several orders of magnitude.
Here, this behaviour cannot be explained by the variance in symmetry as all
systems are (slightly perturbed) symmetric C\textsubscript{720} molecules.
However, we observed that Go-ICP and Go-Permdist can be very fast in two cases:
first, when the optimal orthogonal transformation is close to the identity matrix and second,
if a large amount of noise was added to the one of the particle systems.

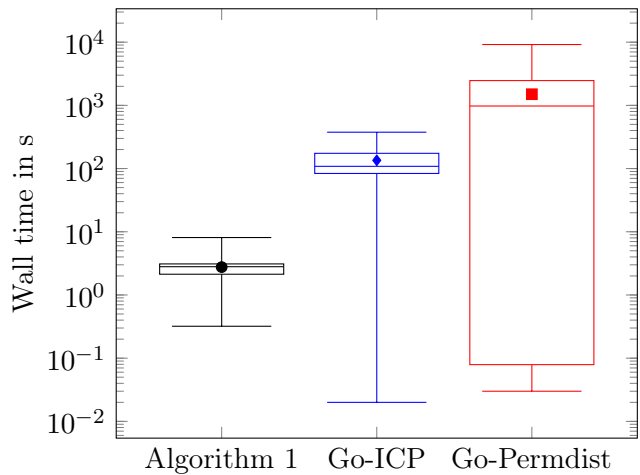
\begin{figure}[ht!]
\begin{center}
\begin{tikzpicture}
\begin{semilogyaxis}[
    boxplot/draw direction = y,
    ylabel={Wall time in s},
    xtick={1,2,3},
    xticklabels = {Algorithm \ref{alg:IRMSD}, Go-ICP, Go-Permdist},
    xticklabel style = {align=center},
    ]

    \addplot+ [
        mark=*,
        black,
        mark options = {solid, fill=black},
        boxplot prepared={
            lower whisker = 0.32,
            lower quartile = 2.13,
            median = 2.8,
            upper quartile = 3.1,
            upper whisker = 8.1,
        },
    ] coordinates {(1, 2.76)};
    \addplot+ [
        mark=diamond*,
        blue,
        mark options = {solid, fill=blue},
        boxplot prepared={
            lower whisker = 0.02,
            lower quartile = 83.7,
            median = 108.5,
            upper quartile = 174.3,
            upper whisker = 376.2,
        },
    ] coordinates {(2, 135)};
    \addplot+ [
        mark=square*,
        red,
        mark options = {solid, fill=red},
        boxplot prepared={
            lower whisker = 0.03,
            lower quartile = 0.079,
            median = 977,
            upper quartile = 2465,
            upper whisker = 9166,
        },
    ] coordinates {(3, 1501)};
\end{semilogyaxis}
\end{tikzpicture}
\end{center}
\caption{Distribution of the wall times in the C720 dataset.}
\label{fig:walltime_c720}
\end{figure}
We also ran our tests on the diamond dataset, choosing the cutoff radius $r=6\textrm{Å}$. 
Similar to the C720 dataset, the diamond dataset contains highly symmetric particle
systems.
Unsurprisingly, the runtime plot in figure \ref{fig:walltime_diamond} looks similar
to the corresponding figure for the C720 dataset.
The runtimes are lower on average, which can be explained by the lower number of particles (159) per
particle system compared to the C720 dataset (720).

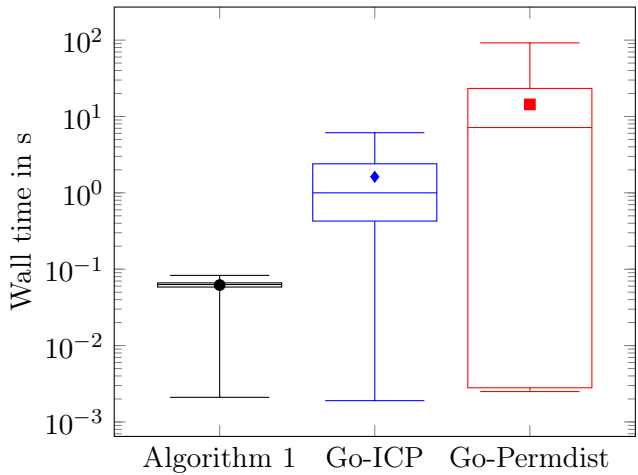
\begin{figure}[ht!]
\begin{center}
\begin{tikzpicture}
\begin{semilogyaxis}[
    boxplot/draw direction = y,
    ylabel={Wall time in s},
    xtick={1,2,3},
    xticklabels = {Algorithm \ref{alg:IRMSD}, Go-ICP, Go-Permdist},
    xticklabel style = {align=center},
    ]

    \addplot+ [
        mark=*,
        black,
        mark options = {solid, fill=black},
        boxplot prepared={
            lower whisker = 0.0021,
            lower quartile = 0.0584,
            median = 0.0630,
            upper quartile = 0.06625,
            upper whisker = 0.083,
        },
    ] coordinates {(1, 0.062)};
    \addplot+ [
        mark=diamond*,
        blue,
        mark options = {solid, fill=blue},
        boxplot prepared={
            lower whisker = 0.0019,
            lower quartile = 0.427,
            median = 1.0,
            upper quartile = 2.4,
            upper whisker = 6.13,
        },
    ] coordinates {(2, 1.62)};
    \addplot+ [
        mark=square*,
        red,
        mark options = {solid, fill=red},
        boxplot prepared={
            lower whisker = 0.0025,
            lower quartile = 0.0028,
            median = 7.17,
            upper quartile = 23.30,
            upper whisker = 92.13,
        },
    ] coordinates {(3, 14.39)};
\end{semilogyaxis}
\end{tikzpicture}
\end{center}
\caption{Distribution of the wall times in the diamond dataset ($r=6$Å).}
\label{fig:walltime_diamond}
\end{figure}

\subsection{Scaling}

To test the scaling of the algorithms with respect to the size of the particle system, we
used different cutoff radii in the diamond dataset to create systems of different sizes.
Again setting $\epsilon=0.2\textrm{Å}$ we looked at the worst-case runtimes as shown in 
figure \ref{fig:scaling}.
\begin{figure}[ht!]
\begin{center}
\begin{tikzpicture}
\begin{loglogaxis}[
    xlabel={Number of particles},
    ylabel={Worst-case wall time in s},
    legend style = {at = {(1, 0)}, anchor = south east},
    xmin=47,xmax=87000,ymin=1e-2, ymax=5e4,
    width=0.7\textwidth,
    height=0.5\textwidth
    ]
    \addplot[-, black, mark=*, style={thick}] table[x index=0, y index=1]{images/diamond_average_irmsd.dat};
    \addlegendentry{Algorithm \ref{alg:IRMSD}}
    \addplot[dotted, blue, mark=diamond*, mark options={solid}, style={thick}] table[x index=0, y index=1]{images/diamond_average_goicp.dat};
    \addlegendentry{Go-ICP}
    \addplot[dashed, red, mark=square*, style={thick}, mark options={solid}] table[x index=0, y index=1]{images/diamond_average_gopermdist.dat};
    \addlegendentry{Go-Permdist}
    \addplot[-, gray, domain=47:87000] {1e-6 * x^2};
    \addlegendentry{$O(n^2)$}
\end{loglogaxis}
\end{tikzpicture}
    \caption{Worst-case wall times for the diamond example.}
    \label{fig:scaling}
\end{center}
\end{figure}
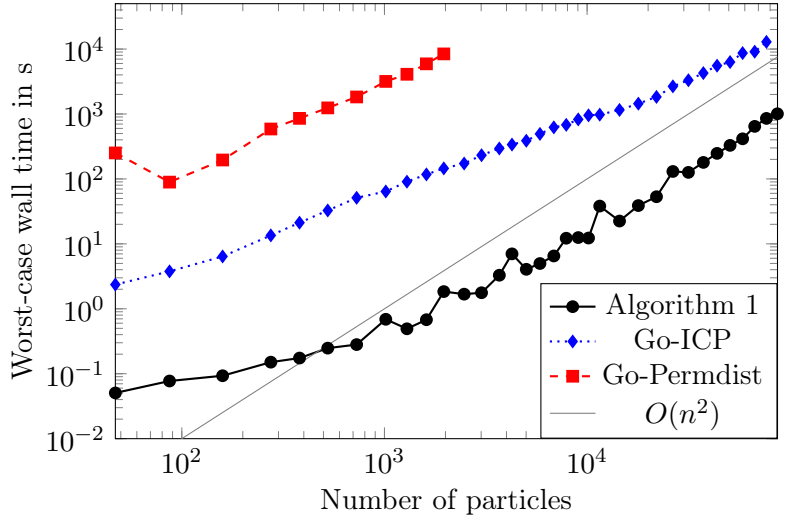

For all tested numbers of particles, the worst-case runtimes of our algorithm are much lower than those of the two
reference algorithms, although the gap between Go-ICP and algorithm \ref{alg:IRMSD} seems
to narrow for large numbers of particles. 
Unexpectedly, the run times of our algorithm seem to scale as $O(n^2)$ with the number of particles
despite the fact that all particles are three-dimensional.
This behaviour can be explained by the setup of the diamond dataset.
To create the particle systems in this dataset, atomic neighbourhoods in a perfect diamond crystal were 
used.
Due to the structure of a diamond crystal, the particles in such a neighbourhood are arranged in shells
around the central particle of the neighbourhood.
This shell-type structure reduces the number of possible permutations, as only particles
with similar distances to the central particle may be matched to each other.

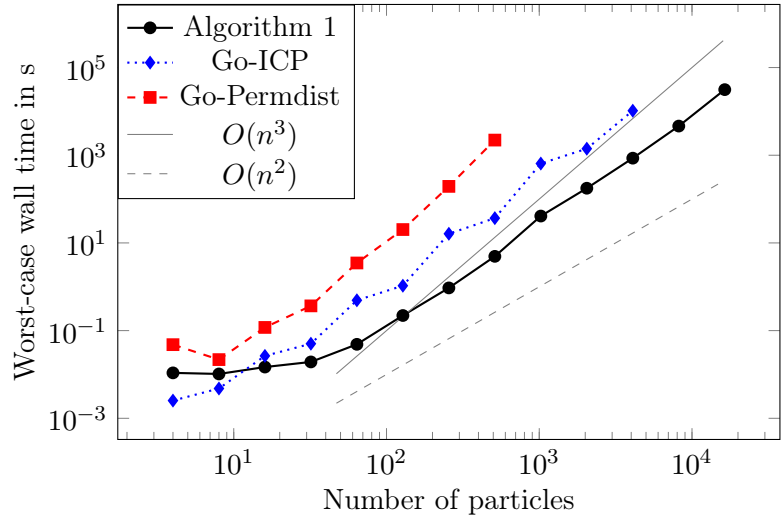
\begin{figure}[ht!]
\begin{center}
\begin{tikzpicture}
\begin{loglogaxis}[
    xlabel={Number of particles},
    ylabel={Worst-case wall time in s},
    legend style = {at = {(0, 1)}, anchor = north west},
    width=0.7\textwidth,
    height=0.5\textwidth
    ]
    \addplot[-, black, mark=*, style={thick}] table[x index=0, y index=1]{images/sphere_worst_irmsd.dat};
    \addlegendentry{Algorithm \ref{alg:IRMSD}}
    \addplot[dotted, blue, mark=diamond*, style={thick}, mark options={solid}] table[x index=0, y index=1]{images/sphere_worst_goicp.dat};
    \addlegendentry{Go-ICP}
    \addplot[dashed, red, mark=square*, style={thick}, mark options={solid}] table[x index=0, y index=1]{images/sphere_worst_gopermdist.dat};
    \addlegendentry{Go-Permdist}

    \addplot[-, gray, domain=47:16000] {1e-7 * x^3};
    \addlegendentry{$O(n^3)$}
    \addplot[dashed, gray, domain=47:16000] {1e-6 * x^2};
    \addlegendentry{$O(n^2)$}
\end{loglogaxis}
\end{tikzpicture}
    \caption{Worst-case wall times for the sphere example.}
    \label{fig:scaling2}
\end{center}
\end{figure}

In the spherical dataset all particles in a particle system lie on the surface of some sphere,
thus there exists only a single shell.
When we use our algorithm on this dataset, we can actually see the $O(n^3)$ scaling that we 
expected in the beginning (see figure \ref{fig:scaling2}). 
Go-ICP and Go-Permdist also seem to scale cubically with the number of particles but are much
slower than algorithm \ref{alg:IRMSD}, except for very small systems.

\section{Conclusion}

We have developed an algorithm that can be used to check particle systems for similarity
and align them if they are similar.
On average, our algorithm is faster than both Go-ICP and Go-Permdist, in some cases 
by several orders of magnitude. 

Compared to Go-ICP and Go-Permdist our algorithm is limited by the upper bound for the
allowed tolerance $\epsilon$ and can thus be used only in some cases.
While we do not think that this will cause problems when used in a database context as 
explained in the introduction,  
it might be worthwhile to extend algorithm \ref{alg:IRMSD} such that it resorts to some form
of branch-and-bound-style permutation search if the tolerance $\epsilon$ is larger 
than the upper bound $\frac{\mu}{2\sqrt{1 + 4d}}$.

In this work we have presented only a comparison approach for particle systems. 
To build a database of particle systems, additional work is needed to design a database
structure that supports the efficient querying of particle systems.

\section*{Appendix}

\subsection{Proofs}

The first proof shows the consequences of choosing the indices $j_1,\dots, j_d$ in algorithm
\ref{alg:IRMSD}:

\begin{theorem}\label{theorem:infinitynorm}
    Let $\mathbf y = (y_1,\dots, y_n)$ be $n$ $d$-dimensional points ($y_i\in\R^d$) and 
    assume that the matrix $\left[y_1\, \dots\, y_n\right]\in\R^{d\times n}$ has
    full (column) rank $d$.
    Let $j_1,\dots, j_d$ be the $d$ indices that maximise
    \begin{equation}
        \max_{j_1,\dots, j_d} \left| \det\left(\left[\begin{array}{ccc} | & & | \\ y_{j_1} & \dots & y_{j_d} \\ | & & | \end{array} \right]\right)\right|. \label{eq:detmax}
    \end{equation}    
    Then each point $y_i$ can be written as
    \begin{equation}
        y_i = \left[\begin{array}{ccc} | & & | \\ y_{j_1} & \dots & y_{j_d} \\ | & & | \end{array} \right] v_i
    \end{equation}
    where $v_i\in\R^d$ fulfils
    \begin{equation}
        \norm{v_i}_\infty\leq 1.
    \end{equation}
    We note that this theorem has already been stated in \cite{mikhalev2015} though without an
    explicit proof.
\end{theorem}

\begin{proof}
    As the matrix $\left[y_1\, \dots\, y_n\right]$ is assumed to have full rank, the determinant in equation \ref{eq:detmax} is
    nonzero and therefore the matrix
    \begin{equation}
        \mathbf{\tilde Y} = \left[\begin{array}{ccc} | & & | \\ y_{j_1} & \dots & y_{j_d} \\ | & & | \end{array} \right]
    \end{equation}
    is invertible. 
    We set $v_i = \mathbf{\tilde Y^{-1}} y_i$ for all indices $i$ and show that it fulfils $\norm{v_i}_\infty\leq 1$.
    Using Cramer's rule, the $k$-th component of $v_i$ is given by
    \begin{equation}
        (v_i)_k = (\mathbf{\tilde Y}^{-1} y_i)_k = \frac{\det \mathbf{\tilde Y}^{(i, k)}}{\det \mathbf{\tilde Y}}.
    \end{equation}
    $\mathbf{\tilde Y}^{(i, k)}$ is the matrix that arises when the $k$-th column of $\mathbf{\tilde Y}$ is replaced by $y_i$.
    As $|\det\mathbf{\tilde Y}| \geq |\det \mathbf{\tilde Y}^{(i, k)}|$ by definition, we finally obtain
    \begin{equation}
        \norm{v_i}_\infty = \max_k |(v_i)_k| = \max_k \left|\frac{\det \mathbf{\tilde Y}^{(i, k)}}{\det \mathbf{\tilde Y}}\right|\leq 1. 
    \end{equation}
\end{proof}

\begin{theorem}\label{theorem:roterror}
    Let $\mathbf{x} = (x_1,\dots, x_n)$ and $\mathbf{y} = (y_1,\dots, y_n)$ be $n$ $d$-dimensional points each 
    ($x_i\in\R^d$ and $y_i\in\R^d$).
    Assume that the $y_1,\dots, y_n$ fulfill the assumptions in theorem \ref{theorem:infinitynorm} and 
    let $j_1,\dots, j_d$ be the corresponding indices that maximise equation \ref{eq:detmax}.
    Also let $i_1,\dots, i_d$ be $d$ distinct indices, $\mathbf R\in\mathbb{O}_d$ and $\mathbf{\tilde R}\in\mathbb{O}_d$ 
    arbitrary orthogonal transformations
    and $\pi\in \mathbb S_n$ an arbitrary permutation that fulfils $\pi_{i_k} = j_k$ ($\forall k=1,\dots, d$).

    Then for any $l=1,\dots, n$, it holds that
    \begin{equation}
        \norm{(\mathbf R - \mathbf{\tilde R}) y_l}_2 \leq \sqrt{d} (\tilde\delta + \delta),
    \end{equation}
    where
    \begin{align}
        \delta &= \sqrt{\sum_{k=1}^d \norm{x_{i_k} - \mathbf Ry_{\pi_{i_k}}}_2^2} \label{eq:delta1}\\ 
        \tilde \delta &= \sqrt{\sum_{k=1}^d \norm{x_{i_k} - \mathbf{\tilde R}y_{\pi_{i_k}}}_2^2}.\label{eq:delta2}
    \end{align}    
\end{theorem}

\begin{proof}
    Using the inequality $\sum_i^d |a_i| \leq \sqrt{d\sum_i^d a_i^2}$ we obtain
    \begin{align*}
        \sum_{k=1}^d &\norm{(\mathbf R - \mathbf{\tilde R}) y_{j_k}}_2 
            = \sum_{k=1}^d \norm{(\mathbf R - \mathbf{\tilde R}) y_{\pi_{i_k}}}_2 \\
            &= \sum_{k=1}^d \norm{(\mathbf R - \mathbf{\tilde R}) y_{\pi_{i_k}} + x_{i_k} - x_{i_k}}_2 \\
            &\leq \sum_{k=1}^d \norm{x_{i_k} - \mathbf{\tilde R} y_{\pi_{i_k}}}_2 + \norm{x_{i_k} - \mathbf R y_{\pi_{i_k}}}_2  \\
            &\leq \sqrt{d\sum_{k=1}^d \norm{x_{i_k} - \mathbf{\tilde R} y_{\pi_{i_k}}}_2^2} \\
            &~~ + \sqrt{d\sum_{k=1}^d \norm{x_{i_k} - \mathbf R y_{\pi_{i_k}}}_2^2 }\\
            &\leq \sqrt{d}(\tilde \delta + \delta).
    \end{align*}
    Using theorem \ref{theorem:infinitynorm} we can write every $y_l$ as
    \begin{displaymath}
        y_l = \mathbf{\tilde Y} v_l = \sum_{k=1}^d (v_l)_k y_{j_k}
    \end{displaymath}
    where $\norm{v_l}_\infty\leq 1$. Thus we obtain for all $l=1,\dots, n$
    \begin{align*}
        \norm{(\mathbf R - \mathbf{\tilde R}) y_l}_2 
            &= \norm{(\mathbf R - \mathbf{\tilde R})\sum_{k=1}^d (v_l)_k y_{j_k}}_2 \\
            &\leq \sum_{k=1}^d |(v_l)_k|\, \norm{(\mathbf R - \mathbf{\tilde R})y_{j_k}}_2 \\
            &\leq \sum_{k=1}^d \norm{(\mathbf R - \mathbf{\tilde R})y_{j_k} }_2 \\
            &\leq \sqrt{d} (\tilde\delta + \delta).
    \end{align*}
\end{proof}

The next theorem is critical to the correctness of algorithm \ref{alg:IRMSD}.
If $\IRMSD(S, \tilde S)\leq\epsilon$ for suitably small $\epsilon$, it shows that the 
permutation $\pi$ that is calculated in line \ref{line:pi} is equal to the globally optimal
permutation $\pi^\ast$ for the indices $i_1,\dots, i_d$ that fulfill $\pi^\ast_{i_k} = j_k$.
As the algorithm tests \emph{all} index combinations $i_1,\dots, i_d$, 
one of them will yield the globally optimal permutation $\pi^\ast$.

\begin{theorem} \label{thm:algproof}
    Let $S = (\mathbf x, \mathbf e)$ and $\tilde S = (\mathbf{\tilde x}, \mathbf{\tilde e})$ be two particle systems
    such that $\IRMSD(S, \tilde S) \leq\epsilon$.
    We denote by $n$ the number of particles in the systems and by $d$ their dimensionality.
    Thus $\mathbf x = (x_1,\dots, x_n)$ consists of $n$ $d$-dimensional points $x_i\in\R^d$ (similar for $\mathbf{\tilde x}$).
    By $\mathbf R^\ast\in\mathbb{O}_d$ and $\pi^\ast\in \mathbb S_n$ we denote the optimal orthogonal matrix and the optimal
    permutation in the calculation of the $\IRMSD$ (see equation \ref{eq:IRMSD}).

    Now assume that $\mathbf{\tilde x} = (\tilde x_1,\dots, \tilde x_n)$ fulfils the assumptions in theorem \ref{theorem:infinitynorm} and 
    that $j_1,\dots, j_d$ are the corresponding indices that maximise equation \ref{eq:detmax}.
    Define the $d$ indices $i_1,\dots, i_k$ implicitly by $\pi^{\ast}_{i_k} = j_k$ ($\forall k=1,\dots, d$)
    and $\mathbf R\in\mathbb{O}_d$ by
    \begin{equation}
        \mathbf R = \argmin_{\mathbf R\in\mathbb{O}_d} \sum_{k=1}^d \norm{x_{i_k} - \mathbf R \tilde x_{j_k} }_2^2.
    \end{equation}
    Let $\mu\in\R$ be the minimal particle distance in $\tilde S$, i.e.
    \begin{equation}
        \mu = \min_{i\neq j} \norm{\tilde x_i - \tilde x_j}_2.
    \end{equation}
    If $\epsilon< \frac{\mu}{2\sqrt{1 + 4d}}$, then $\pi^\ast$ is equal to the 
    permutation $\pi$ that is defined by 
    \begin{equation} \label{eq:piopt}
    \begin{aligned}
        \pi_{i_k} &= j_k && \forall k=1,\dots, d \\
        \pi_k &= \argmin_{l:\ e_k = \tilde e_l} \norm{x_k - \mathbf R\tilde x_l}_2, &&\forall k\not\in\{i_1,\dots, i_d\}.
    \end{aligned}
    \end{equation}
\end{theorem}

\begin{proof}
    By definition, $\pi^\ast_{i_k} = \pi_{i_k}$ for all $k=1,\dots, d$.
    Thus it remains to show that $\pi^\ast_k = \pi_k$ for all $k\not\in\{i_1,\dots, i_d\}$.
    To do this we define 
    \begin{equation}
        \delta = \sqrt{\sum_{k=1}^d \norm{x_{i_k} - \mathbf R \tilde x_{j_k}}_2^2}
    \end{equation}    
    and
    \begin{equation}
        \tilde\delta = \sqrt{\sum_{k=1}^d \norm{x_{i_k} - \mathbf R^\ast \tilde x_{j_k}}_2^2}.
    \end{equation}
    Due to the definition of $\mathbf R$ we have $\delta\leq\tilde\delta\leq\epsilon$.
    Using theorem \ref{theorem:roterror}, we obtain
    \begin{equation}
        \norm{(\mathbf R - \mathbf R^\ast) \tilde x_k}_2 \leq \sqrt{d}(\delta + \tilde\delta) \leq 2\sqrt{d}\tilde\delta
    \end{equation}    
    for all $k=1,\dots, n$.
    For any $k\not\in\{i_1,\dots, i_d\}$ we can write
    \begin{align*}
        \epsilon^2 &\geq \IRMSD(S,\tilde S)^2 = \sum_{l=1}^n \norm{x_l - \mathbf R^\ast \tilde x_{\pi^\ast_l}}_2^2 \\
            &= \sum_{l=1}^d \norm{x_{i_l} - \mathbf R^\ast \tilde x_{j_l}}_2^2 + \sum_{l=1\atop l\not\in\{i_1,\dots, i_d\}}^n \norm{ x_l - \mathbf R^\ast \tilde x_{\pi^\ast_l}}_2^2 \\
            &= \tilde\delta^2 + \sum_{l=1\atop l\not\in\{i_1,\dots, i_d\}}^n \norm{ x_l - \mathbf R^\ast \tilde x_{\pi^\ast_l}}_2^2 \\
            &\geq \tilde\delta^2 + \norm{ x_k - \mathbf R^\ast \tilde x_{\pi^\ast_k}}_2^2.
    \end{align*}
    and thus obtain:
    \begin{equation}
        \norm{x_k - \mathbf R^\ast \tilde x_{\pi^\ast_k}}_2 \leq \sqrt{\epsilon^2 - \tilde\delta^2}.
    \end{equation}
    Putting all things together we obtain for $k\not\in\{i_1,\dots, i_d\}$:
    \begin{equation} \label{eq:final}
        \begin{aligned}
        \norm{x_k - \mathbf R\tilde x_{\pi^\ast_k}}_2 &= \norm{x_k + \mathbf R^\ast \tilde x_{\pi^\ast_k} - \mathbf R^\ast \tilde x_{\pi^\ast_k} - \mathbf R\tilde x_{\pi^\ast_k}}_2 \\
            &\leq \norm{x_k - \mathbf R^\ast \tilde x_{\pi^\ast_k}}_2 + \norm{ (\mathbf R - \mathbf R^\ast) \tilde x_{\pi^\ast_k}}_2 \\
            &\leq \sqrt{\epsilon^2 - \tilde\delta^2} + 2\sqrt{d}\tilde\delta \\
            &\leq \sqrt{1 + 4d}\epsilon, 
        \end{aligned}
    \end{equation}
    where the final inequality can be obtained by maximising $\sqrt{\epsilon^2 - \tilde\delta^2} + 2\sqrt{d}\tilde\delta$
    with respect to $\tilde\delta$.
    Equation \ref{eq:final} shows that $\mathbf R\tilde x_{\pi^\ast_k}$ is located in a sphere of radius $\sqrt{1+4d}\epsilon$ around $x_k$.
    As $\epsilon < \frac{\mu}{2\sqrt{1 + 4d}}$, $\tilde x_{\pi^\ast_k}$ is the only point in $\tilde x$ that lies inside this
    radius and thus $\pi^\ast$ minimises equation \ref{eq:piopt}.
\end{proof}

The following lemma and the subsequent theorem are needed by theorem \ref{thm:compproof}
which is needed to prove the complexity of algorithm \ref{alg:IRMSD}.

\begin{lemma} (Taken from \cite{kressner2019, rademacher2006}) \label{lemma:index}
    Let $\mathbf A\in\R^{d\times d}$ be a square matrix and denote by $\mathbf{A(i)}$ the matrix 
    $\mathbf{A}$ with its $i$-th column removed.
    Then some index $k\in\{1,\dots, d\}$ exist such that 
\begin{displaymath}
    \norm{\mathbf{A} - \mathbf{A(k)} \mathbf{A(k)}^+ \mathbf{A}}_F \leq \sqrt{d}\sigma_d(\mathbf A).
\end{displaymath}
    Here, $\sigma_d(\mathbf A)$ denotes the $d$-th largest (i.e.\ smallest) 
    singular value of $\mathbf A$ and $\mathbf{A(k)}^+$ is the Moore-Penrose pseudo-inverse of 
    $\mathbf{A(k)}$.
\end{lemma}

\begin{theorem}\label{thm:rotdist}
    Let $\mathbf A\in\R^{d\times d}$ be non-singular and $k\in\{1,\dots, d\}$ be an index
    such that 
\begin{displaymath}
    \norm{\mathbf A - \mathbf{A(k)} \mathbf{A(k)}^+ \mathbf A}_F \leq \sqrt{d}\sigma_d(\mathbf A)
\end{displaymath}
    (due to lemma \ref{lemma:index} at least one such index exists).
    Then for each orthogonal matrix $\mathbf U\in\mathbb{O}_d$ at least one of the following inequalities
    holds
\begin{equation} \label{eq:ineq1}
    \norm{\mathbf A_{\ast, k} - \mathbf U \mathbf A_{\ast, k}}_2 \leq (1 +\sqrt{2})\sqrt{d} \norm{\mathbf U \mathbf{A(k)} - \mathbf{A(k)}}_F
\end{equation}
\begin{equation}
    \norm{2\mathbf{A(k)}\mathbf{A(k)}^+ \mathbf A_{\ast, k} - \mathbf{A}_{\ast, k}  - \mathbf U \mathbf A_{\ast, k}}_2 \leq (1 +\sqrt{2})\sqrt{d} \norm{\mathbf U \mathbf{A(k)} - \mathbf{A(k)}}_F.
\end{equation}
\end{theorem}

\begin{proof}
    To simplify the notation, it is assumed w.l.o.g. that $k = d$.
    As $\mathbf A$ is non-singular, its $d$-th column $\mathbf A_{\ast, d}$ can also be written as
\begin{displaymath}
    \mathbf A_{\ast, d} = \mathbf{A(d)} \underbrace{\mathbf{A(d)}^+ \mathbf A_{\ast, d}}_{\eqqcolon \mathbf v} + \mathbf c
\end{displaymath}
    where $\mathbf c$ is from the (one-dimensional) nullspace of $\mathbf{A(d)}^T$.
    Setting $\mathbf{\delta(U)}\coloneqq U\mathbf{A(k)} - \mathbf{A(k)}$ immediately yields
\begin{displaymath}
    U \mathbf{A(d)v} - \mathbf{A(d) v} = \mathbf{\delta(U) v}
\end{displaymath}
    and thus
\begin{displaymath}
    \norm{U \mathbf{A(d)v} - \mathbf{A(d) v}}_2 \leq \norm{\mathbf{\delta(U)}}_F \norm{\mathbf v}_2.
\end{displaymath}
    The norm of $\mathbf v$ can be bounded from above as follows ($e_d$ is the $d$-th unit vector)
\begin{align*}
    \norm{e_d - [\mathbf I\, 0]^T \mathbf v}_2 &= \norm{e_d - \mathbf A^{-1} \mathbf{A(d) A(d)}^+ \mathbf Ae_d}_2 \\
         &=\norm{\mathbf A^{-1} (\mathbf Ae_d - \mathbf{A(d) A(d)}^+ \mathbf Ae_d)}_2 \\ 
         &\leq \norm{\mathbf A^{-1}}_2 \norm{\mathbf A - \mathbf{A(d)A(d)}^+ \mathbf A}_F \norm{e_d}_2 \\
         &\leq \frac{1}{\sigma_d(\mathbf A)} \sqrt{d}\sigma_d(\mathbf A) 
         = \sqrt{d} \\
    \Rightarrow  \norm{e_d - [\mathbf I\, 0]^T \mathbf v}_2^2 &=  \norm{\left(\begin{array}{c}
        \mathbf v \\ 1
    \end{array}\right)}_2^2 \leq d \Rightarrow \norm{\mathbf v}_2^2 \leq d - 1.
\end{align*}
    The next step is to prove a similar inequality for $\norm{\mathbf c - \mathbf{Uc}}_2$
    Due to the properties of orthogonal matrices one can easily verify that $\mathbf{\tilde c}\coloneqq \mathbf{Uc}$ must
    have the same Euclidean norm as $\mathbf c$ and must be in the nullspace of $(\mathbf{UA(d)})^T$. 
    Therefore $\mathbf{\tilde c}$ must fulfil
\begin{align}
    (\mathbf{UA(d)})^T \mathbf{\tilde c} &= (\mathbf{A(d)} + \mathbf{\delta(U)})^T \mathbf{\tilde c}  = 0 \label{eq:rotdist:eq1}\\
    \mathbf{\tilde c}^T \mathbf{\tilde c} &= \mathbf c^T \mathbf c\label{eq:rotdist:eq2}.
\end{align}
    Using once again the fact that $\mathbf{A(d)}$ and $\mathbf c$ span the whole $\R^d$, 
    $\mathbf{\tilde c}$ can be written as $\mathbf{\tilde c} = \mathbf{A(d) w} + a\mathbf c$ for some (yet to be determined)
    $\mathbf w\in\R^{d-1}$ and $a\in\R$.
    Due to equation \eqref{eq:rotdist:eq1}, $\tilde c$ must fulfil
\begin{displaymath}
    \mathbf{A(d)}^T \mathbf{\tilde c} = -\mathbf{\delta(U)}^T\mathbf{\tilde c}
\end{displaymath}
and thus
\begin{displaymath}
    \mathbf{A(d)}^T(\mathbf{A(d) w} + a \mathbf c) 
        = \mathbf{A(d)}^T \mathbf{A(d) w} = -\mathbf{\delta(U)}^T\mathbf{\tilde c}.
\end{displaymath}
    As $\mathbf{A(d)}^T \mathbf{A(d)}$ is invertible, $\mathbf w$ fulfils
\begin{displaymath}
    \mathbf w = - (\mathbf{A(d)}^T \mathbf{A(d)})^{-1} \mathbf{\delta(U)}^T\mathbf{\tilde c}
\end{displaymath}
    and thus (using $\norm{\mathbf{\tilde c}}_2 = \norm{\mathbf c}_2$)
\begin{equation}
\begin{aligned}
   \norm{\mathbf{A(d)w}}_2 &= \norm{-\mathbf{A(d)}(\mathbf{A(d)}^T\mathbf{A(d)})^{-1}\mathbf{\delta(U)}^T\mathbf{\tilde c}}_2 
        = \norm{-{\mathbf{A(d)}^T}^+ \mathbf{\delta(U)}^T\mathbf{\tilde c}}_2 \\
        &\leq \norm{\mathbf{{A(d)}^T}^+}_2 \norm{\mathbf{\delta(U)}^T}_F \norm{\mathbf{\tilde c}}_2 
        = \norm{\mathbf{A(d)}^+}_2 \norm{\mathbf{\delta(U)}}_F \norm{\mathbf c}_2.
\end{aligned} \label{eq:rotdist:eq5}
\end{equation}
    Once again using $\norm{\mathbf c}_2 = \norm{\mathbf{\tilde c}}_2$ one obtains
\begin{align*}
    \mathbf c^T \mathbf c &= \mathbf w^T \mathbf{A(d)}^T \mathbf{A(d) w} + a^2 \mathbf c^T \mathbf c + 2 a\mathbf w^T \mathbf{A(d)}^T \mathbf c\\
        &= \norm{\mathbf{A(d) w}}_2^2 + a^2 \mathbf c^T \mathbf c
\end{align*}
    which is a quadratic equation with respect to $a$ and thus has two solutions $a_+$ and $a_-$, 
    given by 
\begin{displaymath}
    a_\pm = \pm \sqrt{1 - \frac{\norm{\mathbf{A(d)w}}_2^2}{\mathbf c^T \mathbf c}}.
\end{displaymath}
    As $\norm{\mathbf{A(d) w}}_2$ is limited from above (see equation \eqref{eq:rotdist:eq5}), $a_+$ fulfils
\begin{displaymath}
    a_+  = \sqrt{1 - \frac{\norm{\mathbf{A(d)w}}_2^2}{\mathbf c^T \mathbf c}} 
        \geq 1 - \norm{\mathbf{A(d)}^+}_2^2\norm{\mathbf{\delta(U)}}_F^2
\end{displaymath}
    and $a_-\leq -\left(1 - \norm{\mathbf{A(d)}^+}_2^2\norm{\mathbf{\delta(U)}}_F^2\right)$.
    Now $\mathbf{\tilde c}$ is equal to either $\mathbf{A(d) w} + a_+ \mathbf c$
    or $\mathbf{A(d) w} + a_- \mathbf c$. In the first case, $\norm{\mathbf{\tilde c} - \mathbf c}_2^2$
    fulfils
\begin{align*}
    \norm{\mathbf{\tilde c} - \mathbf c}_2^2 &= 
        2 \mathbf c^T \mathbf c - \underbrace{\mathbf c^T \mathbf{A(d)w}}_{=0} - 2 a_+ \mathbf c^T \mathbf c 
        = 2 \mathbf c^T \mathbf c (1 - a_+) 
        \leq  2 \mathbf c^T \mathbf c \norm{\mathbf{A(d)}^+}_2^2 \norm{\mathbf{\delta(U)}}_F^2 \\
        &= 2\norm{\mathbf{A(d)}^+}_2^2 \norm{\mathbf{\delta(U)}}_F^2 \norm{\mathbf c}_2^2 \\
        &= 2\norm{\mathbf{A(d)}^+}_2^2 \norm{\mathbf{\delta(U)}}_F^2 \norm{\mathbf A_{\ast, d} - \mathbf{A(d)}\mathbf{A(d)}^+ \mathbf A_{\ast, d}}_2^2 \\
        &\leq 2\norm{\mathbf{A(d)}^+}_2^2 \norm{\mathbf{\delta(U)}}_F^2 \norm{\mathbf A - \mathbf{A(d)}\mathbf{A(d)}^+ \mathbf A}_F^2 \\
        &\leq 2\frac{1}{\sigma^2_{d-1}(\mathbf{A(d)})} \norm{\mathbf{\delta(U)}}_F^2 d\sigma^2_d(\mathbf A) \\
        &\leq 2 d \norm{\mathbf{\delta(U)}}_F^2
\end{align*}
    and thus equation \eqref{eq:ineq1} holds
\begin{align*}
    \norm{\mathbf{A(d)v} + \mathbf c - \mathbf U(\mathbf{A(d) v} + \mathbf c)}_2 
        &= \norm{-\mathbf{\delta(U) v} + \mathbf c - \mathbf{\tilde c}}_2 
        \leq \norm{\mathbf{\delta(U) v}}_2 + \norm{\mathbf c - \mathbf{\tilde c}}_2 \\
        &\leq \norm{\mathbf{\delta(U)}}_F \sqrt{d} + \sqrt{2d}\norm{\mathbf{\delta(U)}}_F \\
        &= (1 + \sqrt{2})\sqrt{d} \norm{\mathbf{\delta(U)}}_F.
\end{align*}
    If $\mathbf{\tilde c} = \mathbf{A(d) w} + a_- \mathbf c$, the second inequality can be derived
    using $\mathbf{A(d)v} - \mathbf c$ instead of $\mathbf{A(d)v} + \mathbf c$.
\end{proof}

\begin{theorem} \label{thm:compproof}
    Let $(x_1,\dots, x_n)$  be $n$ $d$-dimensional points ($x_i\in\R^d$) such that
    \begin{displaymath}
        \norm{x_i - x_j}_2\geq\mu, \, i\neq j
    \end{displaymath}
    for some $\mu > 0$.
    Furthermore let $\tilde x_1,\dots, \tilde x_d$  be $d$ $d$-dimensional points ($\tilde x_i\in\R^d$)
    that are linearly independent.
    Then the magnitude of the set
    \begin{displaymath}
        I = \left\{
            i\in\{1,\dots, n\}^d:\, \left(i_j\neq i_k\, \forall j\neq k\right) \wedge 
            \left(\exists \mathbf R\in\mathbb{O}_d:\, \sum_{l=1}^d \norm{x_{i_l} - \mathbf R\tilde x_l}_2^2 \leq\epsilon^2\right)
        \right\}
    \end{displaymath}
    is of order $O(2.415^d n^{d-1})$ if $\epsilon < \frac{\mu}{2\sqrt{1 + 4d}}$.
\end{theorem}

\begin{proof}
    To begin with, let $\mathbf{\tilde X} = [\tilde x_1,\dots, \tilde x_n]$ be the matrix 
    that has the vectors $\tilde x_i$ as columns.
    Then $\mathbf{\tilde X}$ fulfils all conditions of lemma \ref{lemma:index} and therefore some
    index $k$ exists such that
\begin{displaymath}
    \norm{\mathbf{\tilde X} - \mathbf{\tilde X(k)\tilde X(k)}^+ \mathbf{\tilde X}}_F \leq \sqrt{d}\sigma_d(\mathbf{\tilde X}).
\end{displaymath}
    W.l.o.g. it is assumed here that $k = d$,
    otherwise one can simply permute the $\tilde x_i$ and $x_i$.
    Then every $i\in I$ can be written as $i = (j, l)$ where $j\in\{1,\dots, n\}^{d-1}$ and
    $l\in\{1,\dots, n\}$.
    If $(j, l)\in I$, $j$ must fulfil
\begin{displaymath}
    \left(j_m\neq j_p\, \forall m\neq p\right) \wedge \left(\exists \mathbf R\in\mathbb{O}_d:\, \sum_{m=1}^{d-1} \norm{x_{j_m} - \mathbf R\tilde x_m}_2^2\leq\epsilon^2
 \right).
\end{displaymath}
    Due to the condition that all entries of $j$ must be distinct, there are only up to $\binom{n}{d-1} = O(n^{d-1})$ possible 
    tuples $j$. 
    It is now shown that for each such $j$, no more than $2\cdot 2.415^d$ possible values $\tilde l$ exist such that 
    $(j, \tilde l)\in I$.

    To prove this, it can easily be observed that for $i = (j, l)\in I$ the relation
\begin{align*}
    &\underbrace{\left\{\mathbf R\in\mathbb{O}_d:\, \sum_{m=1}^{d-1} \norm{x_{j_m} - \mathbf R\tilde x_m}_2^2\leq\epsilon^2\right\}}_{=\mathcal{R}_j} \supseteq \\
    &\quad \left\{\mathbf R\in\mathbb{O}_d:\, \sum_{m=1}^{d-1} \norm{x_{j_m} - \mathbf R\tilde x_m}_2^2 + \norm{x_l - \mathbf R\tilde x_d}_2^2\leq\epsilon^2\right\}
\end{align*}
    holds.
    Furthermore, all $\mathbf R, \mathbf{\tilde R}\in\mathcal{R}_j$ fulfil
\begin{equation} \label{eq:n2:tmp1}
    \norm{\mathbf R\mathbf{\tilde X(d)} - \mathbf{\tilde R\tilde X(d)}}_F \leq 2\epsilon.
\end{equation}
    Therefore, choosing an arbitrary $\mathbf R\in\mathcal{R}_j$ and setting $\mathbf A = \mathbf{R\tilde X}$
    in theorem \ref{thm:rotdist} yields two vectors $\mathbf y\in\R^d$ and $\mathbf{\tilde y}\in\R^d$
    such that $\forall \mathbf{\tilde R}\in\mathcal{R}_j$ either $\mathbf y$ or $\mathbf{\tilde y}$ fulfil
\begin{align*}
    \norm{\mathbf  y - \mathbf{\tilde R\tilde X}_{\ast, d}}_2 & = \norm{\mathbf y - \mathbf{\tilde R R}^T (\mathbf{R \tilde X}_{\ast, d})}_2 
    \leq (1 + \sqrt{2})\sqrt{d} \norm{\mathbf{R\tilde X}(d) - \mathbf{\tilde RR}^T (\mathbf R \mathbf{\tilde X(d)})}_F \\
    &= (1 + \sqrt{2})\sqrt{d} \norm{\mathbf{R\tilde X(d)} - \mathbf{\tilde R\tilde X(d)}}_F.
\end{align*}
    Due to equation \eqref{eq:n2:tmp1} this leads to either
\begin{displaymath}
    \norm{\mathbf y - \mathbf{\tilde R}\tilde x_d}_2 \leq 2 (1 + \sqrt{2})\sqrt{d}\epsilon
\end{displaymath}
    or
\begin{displaymath}
    \norm{\mathbf{\tilde y} - \mathbf{\tilde R}\tilde x_d}_2 \leq 2 (1 + \sqrt{2})\sqrt{d}\epsilon
\end{displaymath}
    Due to these inequalities, each feasible $\mathbf{\tilde R}\in\mathcal{R}_j$ transforms $\tilde x_d$ into a 
    sphere of radius $2(1 + \sqrt{2})\sqrt{d}\epsilon$ around either $\mathbf y$ or $\mathbf{\tilde y}$.
    Thus, the corresponding $x_d$ must also lie in one of these spheres.    
    As the $x_i$ were assumed to fulfil $\norm{x_i - x_j}_2\geq\mu$, only a limited number of $x_i$-s 
    can lie in these spheres. 
    As the diameter of each sphere is no more than
\begin{displaymath}
    \frac{2\cdot 2(1 + \sqrt{2})\sqrt{d}\epsilon}{\mu} <
        \frac{2\cdot 2(1 + \sqrt{2})\sqrt{d}\epsilon}{2\sqrt{1 + 4d}\epsilon} =
        \frac{2(1 + \sqrt{2})\sqrt{d}}{\sqrt{d}\sqrt{\frac{1}{d} + 4}}
        \leq \frac{2 + 2\sqrt{2}}{\sqrt{4}} \leq 2.415
\end{displaymath}
    times larger than $\mu$ -- and thus independent of $n$ -- the maximal number of points $x_i$ that fit
    in such a sphere cannot be larger than $2.415^d$.

    Thus if $(j, l)\in I$, for each $j$ there are no more than $2\cdot 2.415^d$ possible $\tilde l\in\{1,\dots,n\}$
    such that $(j, \tilde l)$ is also in $I$.
\end{proof}

\subsection{Go-Permdist lower bound}
\label{sec:gopermdist}
The Go-Permdist algorithm in \cite{griffiths2017} tries to find the optimal $\IRMSD$ by decomposing
the space of all rotations into \emph{rotation cubes} $C(\mathbf v, \theta_B) = \{\mathbf r\in[-\pi,\pi]^3:\, \norm{\mathbf r - \mathbf v}_\infty\leq\theta_B\}$.
Each vector $\mathbf r$ in a rotation cube represents a rotation of angle  $\norm{\mathbf r}_2$ around the rotation axis
$\mathbf r / \norm{\mathbf r}_2$.
For simplicity we will write $\mathbf r(\mathbf x)$ to denote a vector $\mathbf x$ that is rotated by $\mathbf r$.
In order to work, the Go-Permdist algorithm needs a lower bound $LB(\mathbf v,\theta_B)$ that depends
only on $\mathbf v$ and $\theta_B$ such that
\begin{equation}
    LB(\mathbf v, \theta_B) \leq \cos(\angle(\mathbf v(\mathbf x), \mathbf r(\mathbf x))), \quad \forall \mathbf r\in C(\mathbf v, \theta_B), \mathbf x\in\R^3.
\end{equation}
In equation 19 of the cited paper such a lower bound is derived that unfortunately does not hold true in all
cases.
As a simple counter example $\mathbf v = (-0.25\pi, 0.25\pi, 0.25\pi)$, $\mathbf r = (-0.48\pi, 0.02\pi, 0)$,
$\mathbf x = (0.12\pi, -0.14\pi, 0)$ and $\theta_B = 0.5\pi$ can be chosen.
As the problem seems to be just a missing factor $\sqrt{3}$, 
adding it to equation 19 in the original paper results in the corrected lower bound
\begin{equation}
    LB(\mathbf v,\theta_B) = \cos\left(\min\left[\pi, \sqrt{3}\frac{\theta_B}{2}\right] \right).
\end{equation}

\section*{Acknowledgement}
This work was funded in parts by the Bundesministerium für Bildung und Forschung (BMBF)
as a part of the Eurostars project E!9389 MultiModel. \\
The main algorithm and all related proofs originate from the PhD thesis of 
one of the authors \cite{bulin2021}.

\section*{Conflict of interest disclosure}

The \emph{Fraunhofer Institute for Algorithms and Scientific Computing} offers commercial 
software products that include the described comparison procedure.

\bibliographystyle{unsrt}
\bibliography{Report}

\end{document}